\documentclass[%
 reprint,
superscriptaddress,
nofootinbib,
 amsmath,amssymb,
pra,
longbibliography
]{revtex4-1}
\usepackage[utf8]{inputenc}
\usepackage{amsmath,amsthm}
\usepackage{dcolumn}
\usepackage{bm}
\usepackage[colorlinks=true,citecolor=blue,urlcolor=blue]{hyperref}
\usepackage[pdftex]{graphicx}
\usepackage{braket}
\usepackage{color}
\usepackage[usenames,dvipsnames,svgnames,table]{xcolor}

\newcommand{\be}{\begin{equation}}
\newcommand{\ee}{\end{equation}}
\newcommand{\ba}{\begin{eqnarray}}
\newcommand{\ea}{\end{eqnarray}}
\newcommand{\ketbra}[2]{|#1\rangle \langle #2|}
\newcommand{\tr}{\operatorname{Tr}}

\newtheorem{Lemma}{Lemma}

\newtheorem{Result}{Result}
\begin{document}

\title{One-sided Device-independent Self-testing of any  Pure Two-qubit Entangled State}

\author{Suchetana Goswami}
\email{suchetana.goswami@gmail.com}
\affiliation{S.   N.  Bose  National Centre  for Basic  Sciences, Salt
  Lake,    Kolkata    700    106,   India}  
\author{Bihalan Bhattacharya}
\email{bihalan@gmail.com}
\affiliation{S.   N.  Bose  National Centre  for Basic  Sciences, Salt
  Lake,    Kolkata    700    106,   India}  
  \author{Debarshi      Das}
  \email{dasdebarshi90@gmail.com}
\affiliation{Centre for Astroparticle
  Physics and Space Science (CAPSS),  Bose Institute, Block EN, Sector
 V,  Salt Lake,  Kolkata 700  091, India} 
 \author{Souradeep Sasmal}  
\affiliation{Centre for Astroparticle
  Physics and Space Science (CAPSS),  Bose Institute, Block EN, Sector
  V,  Salt Lake,  Kolkata 700  091, India}
   \author{C.  Jebaratnam} 
   \email{jebarathinam@bose.res.in}
   \affiliation{S. N. Bose National Centre for
  Basic Sciences, Salt Lake, Kolkata 700 106, India}
    \author{A.    S.     Majumdar}
\email{archan@bose.res.in} \affiliation{S. N. Bose National Centre for
  Basic Sciences, Salt Lake, Kolkata 700 106, India} 
  
\date{\today}

\begin{abstract}

We consider the problem of $1$-sided device-independent self-testing of any pure entangled two-qubit state
based on steering inequalities which certify the presence of quantum steering. In particular, we note that 
in the $2-2-2$ steering scenario (involving $2$ parties, $2$ measurement settings per party, $2$ outcomes per measurement setting), 
the maximal violation of a fine-grained steering inequality can be used 
to witness certain extremal steerable correlations,
which certify all pure two-qubit entangled states.  We demonstrate that the violation of the analogous CHSH inequality of steering or nonvanishing value
of a quantity constructed using a correlation function called mutual predictability, 
together with the maximal violation of the fine-grained steering inequality  can be used to self-test any pure  entangled two-qubit state in a $1$-sided device-independent way.
\end{abstract}


\maketitle

\section{Introduction}

Quantum information processing utilizes three types of quantum inseparabilities \cite{EPR35, Schrodinger,Rei89,WJD07,JWD07,Bel64} for multipartite systems. These are entanglement \cite{HHH+09}, steering \cite{CS17} and Bell-nonlocality \cite{BCP+14}. 
Entangled states were first introduced in the context of the famous Einstein-Podolsky-Rosen (EPR) paradox \cite{EPR35}. Later on, in the same year Schr\"{o}dinger introduced the concept of  steering~\cite{Schrodinger} in response to the EPR paradox. EPR (or quantum) steering, as pointed out by Schr\"{o}dinger, occurs in the scenario where a bipartite state shared between, say, Alice and Bob is entangled and Alice prepares different ensembles of quantum states for Bob by performing measurements on her subsystem. 

More precisely, the concept of steering is ingrained in the fact that if the bipartite state (shared between Alice and Bob) is steerable from Alice to Bob then Alice can convince Bob by some local operation and classical communication that the state they are sharing is entangled (while Bob does not trust Alice). Basically, a bipartite quantum state exhibits steering if the conditional states prepared on one side by performing measurements on the other side cannot be modeled by a description known as local hidden state (LHS) model (i. e., Bob holds a grand ensemble whose elements are classically correlated with the outcomes of Alice's measurement) \cite{WJD07, JWD07}.  In \cite{WJD07,JWD07}, the authors have demonstrated that Bell-nonlocal states form a strict
subset of steerable states which also form a strict subset of entangled states.

Based on two assumptions, viz. no signalling and the validity of quantum theory, the device-independent (DI) certification of quantum devices is a relevant research direction in quantum information as well as in quantum foundations \cite{Sca12}. The DI approach has several applications, for example, in random number certification \cite{PAM+10}, cryptography \cite{MAG06}, testing the dimension of Hilbert Spaces \cite{BPA+08}. In Ref. \cite{MY04}, a DI scheme called self-testing was proposed to certify a Bell state (maximally entangled two-qubit state) up to local isometries.
Moreover, it has been shown that nonlocal correlations which are extremal \footnote{An extremal quantum correlation in a given Bell scenario cannot be decomposed as a convex mixture of the other quantum correlations in that given Bell scenario. 
Note that, there  exists extremal quantum correlations which do not violate any Bell's inequality maximally in that Bell scenario \cite{GKW+18}.} can be used for self-testing as these extremal correlations can only be achieved by performing particular measurements on a unique pure quantum state (up to some local isometry) \cite{CGS17}. 
Thus, by observing these extremal quantum correlations, it is possible to identify the entangled state without any assumption on the physical systems, measurements or even on the dimension of the relevant Hilbert Space. 

In Ref. \cite{MYS12}, the first criterion for robust self-testing of a singlet state (maximally entangled bipartite qubit state) in the DI scenario was proposed  using the maximal violation of Bell-CHSH (Bell-Clauser-Horne-Shimony-Holt) inequality \cite{CHS+69}. 
In Ref. \cite{AMP12}, a family of Bell inequalities called tilted Bell-CHSH inequality was studied. A family of extremal nonlocal correlations  which can be simulated by a pure two-qubit entangled state can be identified by the maximal violation of the tilted Bell-CHSH inequality.
In Ref. \cite{YN13, BP_15}, it has been shown that any pure two-qubit entangled state can be self-tested in a fully DI way by using these extremal correlations. In Refs. \cite{MY04,RZS12}, criteria for DI certification of quantum system were proposed without using Bell inequalities. In Ref. \cite{CGS17}, it has been shown  that any pure two-qudit entangled state can be self-tested in the Bell scenario where Alice and Bob perform three and four $d$-outcomes measurements on their respective sides. 

Steering inequalities \cite{CJW+09,ZHC16} which are analogous to Bell inequalities are used to certify the presence 
of steering.
The violation of a steering inequality certifies the presence of entanglement in a one-sided device-independent ($1$SDI) way. 
This has implication for quantum information processing in which quantum steering has been used as a resource for $1$SDI quantum key distribution
and randomness generation \cite{BCW+12,LTB+14}. 
It should be noted that it is easier and more cost effective to implement these 1SDI tasks than to implement the completely DI tasks in laboratory. It is thus very relevant and important to study the self-testing problem in the 1SDI framework. Recently, 1SDI self-testing of maximally entangled two-qubit state  has been proposed \cite{SH16,GWK17}. In this context, the maximal violation of a linear steering inequality \cite{CJW+09} was shown to self-test the maximally entangled two-qubit state in a 1SDI way. Moreover, self-testing via quantum steering was shown to provide certain advantages over DI self-testing.   

In this work, we are interested in the problem of self-testing of any pure two-qubit entangled state in the 1SDI scenario. For this purpose we consider two steering inequalities, {\it viz.} the fine grained inequality (FGI) \cite{PKM14}, whose maximum violation certifies that the shared state is pure two-qubit entangled state, and the analogous CHSH inequality for steering \cite{CFF+15}. We demonstrate that the violation of the analogous CHSH inequality of steering 
together with the maximal violation of the fine-grained steering inequality  can be used to self-test any pure  entangled two-qubit state in a $1$SDI way. We further propose another scheme for $1$SDI self-testing of any  pure two-qubit entangled state in which the non-vanishing value
of a quantity constructed using a correlation function called ``mutual predictability" together with the maximal violation of the fine-grained steering inequality is used.  

\section{Backdrop}
\subsection{Quantum steering}
Let us consider the following steering scenario. Two spatially separated parties, say Alice and Bob, share an unknown quantum system $\rho_{AB}\in \mathcal{B}(\mathcal{H}_A \otimes \mathcal{H}_B)$,  where $\mathcal{B}(\mathcal{H}_A \otimes \mathcal{H}_B)$ stands for the set of all density operators acting on the Hilbert space $\mathcal{H}_A \otimes \mathcal{H}_B$, with the Hilbert space dimension of Alice's subsystem is arbitrary and the Hilbert space dimension of Bob's subsystem is fixed. Alice performs a set of uncharacterized measurements (i. e., Alice's measurement operators   $\{M_{a|x}\}_{a,x}$ are elements of unknown positive operator valued measurements (POVM); $M_{a|x} \geq 0$ $\forall a, x$; and $\sum_{a} M_{a|x} = \mathbb{I}$ $\forall x$) on her part of the shared bipartite system $\rho_{AB}$  to prepare the set of conditional states on Bob's side. Here $x=0, 1, 2, ...$ denotes Alice's choice of measurement setting and $a = 0, 1, 2, ...$ denotes outcome of Alice's 
 measurement. Bob can do state tomography to determine the conditional states prepared on his side by Alice. Such a steering scenario is called $1$-sided device-independent (1SDI) since Alice's measurements are treated as black-box measurements. The steering scenario is characterized by an assemblage $\{\sigma_{a|x}\}_{a,x}$ \cite{Pus13} which is the set of unnormalized conditional states on Bob's side.
Each element in the assemblage  is given by $\sigma_{a|x}=p(a|x)\rho_{a|x}$,  where $p(a|x)$ is the conditional probability of getting the outcome $a$ when Alice performs the measurement $A_x$; $\rho_{a|x}$ is the normalized conditional state on Bob's side. Quantum theory predicts that all valid assemblages should satisfy the following criteria:
\begin{equation}
\sigma_{a|x}= \tr_A [( M_{a|x} \otimes \openone) \rho_{AB}] \hspace{0.5cm} \forall \sigma_{a|x} \in \{\sigma_{a|x}\}_{a,x}.
\end{equation}

In the above scenario, Alice demonstrates steerability to Bob
if the assemblage does not have a local hidden state (LHS) model, i.e., if for all $a$, $x$, there
is no decomposition of $\sigma_{a|x}$ in the form,
\begin{equation}
\sigma_{a|x}=\sum_\lambda p(\lambda) p(a|x,\lambda) \rho_\lambda,
\label{newlhs}
\end{equation}
where $\lambda$ denotes classical random variable which occurs with probability 
$p(\lambda)$; $\rho_{\lambda}$
are called local hidden states which satisfy $\rho_\lambda\ge0$ and
$\tr\rho_\lambda=1$.

We now consider a steering scenario in which the trusted party, Bob,  performs a set of POVMs with elements $\{M_{b|y}\}_{b,y}$ ($M_{b|y} \geq 0$ $\forall b, y$; and $\sum_{b} M_{b|y} = \mathbb{I}$ $\forall y$) 
on the conditional states prepared by Alice's unknown POVMs turning the assemblage  
$\{\sigma_{a|x}\}_{a,x}$ into  measurement correlations $p(ab|xy)$, where $p(ab|xy)$ = $\tr ( M_{b|y} \sigma_{a|x} )$. Here $y=0, 1, 2, ...$ denotes Bob's choice of measurement setting and $b = 0, 1, 2, ...$ denotes outcome of Bob's measurement. The correlation $p(ab|xy)$ detects steerability from Alice to Bob, iff it does not have a decomposition as follows \cite{WJD07, JWD07}: 
\begin{equation}
p(ab|xy)= \sum_{\lambda} p(\lambda) p(a|x,\lambda) p(b|y, \rho_\lambda) \hspace{0.3cm} \forall a,x,b,y; 
\label{LHV-LHS}
\end{equation}
where, $\sum_{\lambda} p(\lambda) = 1$, and $p(a|x, \lambda)$ denotes an arbitrary probability distribution arising from local hidden variable (LHV) $\lambda$ ($\lambda$ occurs with probability $p(\lambda)$) and $p(b|y, \rho_{\lambda}) $ denotes the quantum probability of outcome $b$ when measurement $B_y$ is performed on local hidden state (LHS) $\rho_{\lambda}$. Hence, the box $p(ab|xy)$ will be called steerable correlation iff it does not have a LHV-LHS model. Steerable correlation is certified through the violation of a steering inequality \cite{CJW+09}.


Various quantifiers of EPR steering have been proposed till date and for the purpose of the present study, now we discuss in brief about steerable weight \cite{SNC14},
which is a convex steering monotone \cite{GA15}.
Consider the following decomposition of an arbitrary assemblage $\{ \sigma_{a|x} \}_{a,x}$:
\begin{equation}
\sigma_{a|x} = p_s \sigma^{S}_{a|x} + (1-p_s) \sigma^{US}_{a|x} \hspace{0.5cm} \forall a, x,
\end{equation}
where $0 \leq p_s \leq 1$, $\sigma^{S}_{a|x}$ is a steerable assemblage and $\sigma^{US}_{a|x}$ is an element of unsteerable assemblage having LHS model. The weight of the steerable part $p_s$ minimized over all possible decompositions of the given assemblage $\{ \sigma_{a|x} \}_{a,x}$ gives the steerable weight $SW(\{ \sigma_{a|x} \}_{a,x})$ of that assemblage.

\subsection{Self-testing via quantum steering}
DI self-testing of quantum states through the violation of a Bell inequality occurs only for pure entangled states because it requires 
the observation of an extremal nonlocal correlation \cite{GKW+18}. 
Therefore, in the self-testing problem, certifying a particular pure entangled state is of interest. 

Suppose Alice and Bob want to self-test a particular pure entangled state $\ket{\tilde{\psi}}_{AB} \in \mathcal{H}'_A \otimes \mathcal{H}_B$
from the steerable assemblage arising from the state $\ket{\psi}_{AB} \in \mathcal{H}_A \otimes \mathcal{H}_B$ and measurement operators $\{M_{a|x}\}_{a,x}$ on Alice's side in the aforementioned $1$SDI scenario. Then the assemblage self-tests the pure entangled state 
$\ket{\tilde{\psi}}_{AB}$ if there exists an isometry on Alice' side $\Phi$: $\mathcal{H}_A  \rightarrow \mathcal{H}_A \otimes \mathcal{H}'_A$
such that
\begin{eqnarray}
 \Phi(\ket{\psi}_{AB})&=&\ket{junk}_A \otimes \ket{\tilde{\psi}}_{AB}, \nonumber \\
 \Phi(M_{a|x} \otimes \openone \ket{\psi}_{AB})&=&\ket{junk}_A \otimes (\tilde{M}_{a|x} \otimes \openone ) \ket{\tilde{\psi}}_{AB},
\end{eqnarray}
where $\ket{junk}_A \in \mathcal{H}_A$ and $\{\tilde{M}_{a|x}\}_{a,x}$ are the measurement operators acting on the Hilbert space $\mathcal{H}'_A$. In Ref. \cite{SH16}, self-testing of maximally entangled two-qubit state based on 
the steerable assemblage arising from the two-setting steering scenario was proposed.

Analogous to the DI self-testing scheme based on the maximal violation of a Bell inequality, the measurement correlations $p(ab|xy)$ = $Tr ( \Pi_{b|y} \sigma_{a|x} )$ arising from the assemblage can also be used to self-test the particular pure entangled state. In Refs. \cite{SH16,GWK17}, self-testing of the maximally entangled two-qubit state based on 
the maximal steerable correlation was proposed through the maximal violation of a steering inequality. That is, it was shown 
that the maximal violation of the linear steering inequality \cite{CJW+09},
\begin{equation}\label{LSI}
\braket{A_0  \sigma_z} + \braket{A_1  \sigma_x} \le \sqrt{2},
\end{equation}
self-tests a maximally entangled two-qubit state. Here $\langle A_x B_y \rangle =  \sum_{a,b} (-1)^{a\oplus b} p(ab|xy)$ with $B_y$ being equal to $\sigma_z$ or $\sigma_x$.

\section{A $1$SDI Self-testing Scheme for any pure bipartite qubit entangled state}
Here we consider a steering scenario in which Alice performs two black-box dichotomic measurements 
and Bob performs two qubit measurements in mutually unbiased bases for self-testing any pure two-qubit
entangled state in a $1$SDI way. For this steering scenario, a necessary and sufficient condition for quantum steering 
in the form of steering inequality has been proposed 
in Ref. \cite{CFF+15}. This steering inequality is given by
\begin{align}
&\sqrt{\langle (A_{0} + A_{1}) B_{0} \rangle^2 + \langle (A_{0} + A_{1}) B_{1} \rangle^2 } \nonumber \\
&+\sqrt{\langle (A_{0} - A_{1}) B_{0} \rangle^2 + \langle (A_{0} - A_{1}) B_{1} \rangle^2 } \leq 2, 
\label{chshst}
\end{align}
where $\langle A_x B_y \rangle =  \sum_{a,b} (-1)^{a\oplus b} p(ab|xy)$. This inequality is called the analogous CHSH inequality for quantum steering. We will call this analogous CHSH inequality for steering as CFFW (Cavalcanti-Foster-Fuwa-Wiseman) inequality afterwards.

Our self-testing scheme for certifying any pure bipartite qubit entangled state is based on the violation 
of the CFFW inequality with the maximal violation of another steering inequality, {\it i.e.}, the fine-grained 
inequality (FGI), proposed in Ref. \cite{PKM14}.
The FGI for steering is given by,
\begin{equation}
P(b_{B_{0}}\mid a_{A_{0}})+P(b_{B_{1}}\mid a_{A_{1}})\leq 1+\frac{1}{\sqrt{2}}.
\label{FGI}
\end{equation}
$P(b_{B_{0}}\mid a_{A_{0}})$ is the probability of obtaining the outcome $b$ when Bob performs $B_{0}$ measurement given that Alice has obtained the outcome $a$ by performing the measurement $A_{0}$; $P(b_{B_{1}}\mid a_{A_{1}})$ is the probability of obtaining the outcome $b$ when Bob performs $B_{1}$ measurement given that Alice has obtained the outcome $a$ by performing the measurement $A_{1}$. 
Interestingly, we will now demonstrate that the maximum   violation of the FGI by a shared two-qubit state is achieved if and only if the shared state is any pure (maximally or non-maximally) entangled two-qubit state.

\begin{Lemma}
Suppose the trusted party, Bob, performs projective qubit measurements in mutually unbiased bases (as we will consider CFFW inequality for steering later) corresponding to the operators $B_{0}= \ket{0}\bra{0}-\ket{1}\bra{1}$ and  $B_{1}=\ket{+}\bra{+}-\ket{-}\bra{-}$, where $\{|0\rangle, |1\rangle\}$ is an orthonormal basis and $\{|+\rangle, |-\rangle\}$ is another orthonormal basis given by, $\ket{+}=\frac{1}{\sqrt{2}}(\ket{0}+\ket{1})$ and $\ket{-}=\frac{1}{\sqrt{2}}(\ket{0}-\ket{1})$. Then the correlation violates FGI maximally if and only if the two-qubit state has the form,
\begin{equation}
\label{state1}
\ket{\psi(\theta)}=\cos{\theta} \ket{00}+\sin{\theta} \ket{11} \hspace{0.5cm}
0 < \theta < \frac{\pi}{2} 
\end{equation}
up to local unitary transformations and Alice performs projective measurements corresponding to the two operators given by, $A_{0}= \ket{0}\bra{0}-\ket{1}\bra{1}$ and $A_{1}=\cos{2\theta}(\ket{0}\bra{0}-\ket{1}\bra{1})+\sin{2\theta}(\ket{+}\bra{+}-\ket{-}\bra{-})$ (or their local unitary equivalents).
\end{Lemma}
\begin{proof}
It can be checked by simple calculation that, for $B_0$, $B_1$ mentioned above, if Alice and Bob share the state $\ket{\psi(\theta)}$ given by Eq.(\ref{state1}) and Alice performs  projective measurements corresponding to the operators $A_{0}= \ket{0}\bra{0}-\ket{1}\bra{1}$ and $A_{1}=\cos{2\theta}(\ket{0}\bra{0}-\ket{1}\bra{1})+\sin{2\theta}(\ket{+}\bra{+}-\ket{-}\bra{-})$ (or their local unitary equivalents), then the value of left hand side of FGI is $2$. This numerical value $2$ is the maximum violation of FGI since the algebraic maximum of left hand side of FGI is $2$.

Violation of any steering inequality implies that the shared state is steerable and, hence, entangled. Therefore, the shared bipartite qubit state giving rise to the maximum violation of FGI is either a pure or a mixed entangled state. Note that the left hand side of FGI is the sum of two conditional probabilities and the magnitude of its maximum quantum violation is $2$. Hence, maximum quantum violation of FGI implies that each of the two conditional probabilities appearing in the left hand of FGI given by Eq.(\ref{FGI}) is equal to $1$, i.e. $P(b_{B_{0}}\mid a_{A_{0}})$ = $1$ and $P(b_{B_{1}}\mid a_{A_{1}})$ = $1$. This implies that, by performing measurements of the observables corresponding to the operators $A_0$ and $A_1$ on her particle, Alice can predict with certainty the outcomes of Bob's two different measurements of the two observables corresponding to the operators $B_0$ and $B_1$, respectively, without interacting with Bob's particle, where $B_0$ and $B_1$ are two mutually unbiased qubit measurements as described earlier. 
This is nothing but the EPR paradox\footnote{EPR paradox occurs when Alice’s pair of local quantum measurements prepare two different set of quantum states at Bob’s end which are eigenstates of two noncommuting observables.} \cite{EPR35}. 
Suppose, ${\rho}_{0|0}$ and $\rho_{1|0}$ denote the normalised conditional states at Bob's end when Alice gets outcome $a=0$ and $a=1$, respectively, by performing the measurement $A_0$. The states $\rho_{0|0}$ and $\rho_{1|0}$ should be the eigenstates of the operator $B_0$ as maximum quantum violation of FGI implies Bob's conditional probability $P(b_{B_{0}}\mid a_{A_{0}})$ = $1$ for $a=0$ and for $a=1$. Again, suppose, $\rho_{0|1}$ and $\rho_{1|1}$ denote the normalised conditional states at Bob's end while Alice gets the outcome $a=0$ and $a=1$, respectively, by performing the measurement $A_1$. Following similar arguments, it can be shown that the states $\rho_{0|1}$ and $\rho_{1|1}$ should be the eigenstates of the operator $B_1$. Hence, all the four conditional states at Bob's side $\rho_{0|0}$, $\rho_{1|0}$, $\rho_{0|1}$ and $\rho_{1|1}$ are pure. 
If the shared state between Alice and Bob is a pure entangled state, then it has been shown that the four conditional states at Bob's side are pure \cite{SharpContra}. Now in the following we prove that these pure steerable assemblages can not be obtained from any mixed state, hence showing that the maximal violation of FGI can be obtained only from a pure entangled state.

Let us assume that $\sigma_{0|0}$, $\sigma_{1|0}$, $\sigma_{0|1}$ and $\sigma_{1|1}$ denote the elements of assemblage prepared at Bob's side which corresponds to maximum violation of FGI. Each element $\sigma_{a|x}$ of the assemblage is related to the normalised conditional state $\rho_{a|x}$ at Bob's side through the relation given by, $\sigma_{a|x} = p(a|x) \rho_{a|x}$, where $p(a|x)$ is the conditional probability of getting the outcome $a$ when Alice performs the measurement $A_x$; $x \in \{0, 1\}$; $a \in \{0, 1\}$. Since each of the conditional states $\rho_{a|x}$ are pure, they cannot be expressed as convex mixture of two different states. Moreover, these conditional states are associated with the steerable assemblage $\{ \sigma_{a|x} \}_{a,x}$ giving rise to maximum violation of FGI. Hence, steerable weight of the assemblage $\{ \sigma_{a|x} \}_{a,x}$ giving rise to maximum violation of FGI must be $1$.

According to Lewenstein-Sanpera decomposition, any bipartite qubit state $\rho$ has a unique decomposition in the form \cite{LS98} 
\begin{equation}
\rho = \mu \rho^{ent}_{pure} + (1-\mu) \rho^{sep},
\label{unique}
\end{equation}
where $0 \leq \mu \leq 1$, $\rho^{ent}_{pure}$ is a bipartite qubit pure entangled state and $\rho^{sep}$ is a bipartite qubit separable state. Here, $\mu = 1$ implies that $\rho$ is a pure entangled state and $\mu = 0$ implies that $\rho$ is separable. So, for any bipartite qubit mixed entangled state $\rho^{m}$, $\mu \neq 1$ and $\mu \neq 0$, i. e., $0 < \mu < 1$. Consider $\{ \sigma^m_{a|x} \}_{a, x}$ denotes an arbitrary assemblage produced from $\rho^{m}$ when Alice performs measurements $\{ M_{a|x} \}_{a, x}$. Here $\sigma^m_{a|x} = \tr_{A} [( M_{a|x} \otimes \openone) \rho^m ]$, $\forall \sigma^m_{a|x} \in \{ \sigma^m_{a|x} \}_{a, x}$. Since $\rho^m$ can always be expressed in the form given in Eq. (\ref{unique}),
\begin{equation}
\rho^m = \tilde{\mu} \tilde{\rho}^{ent}_{pure} + (1- \tilde{\mu}) \tilde{\rho}^{sep},
\label{uniquemixed}
\end{equation} 
with $0 < \tilde{\mu} < 1$, we have for all $a$ and $x$
\begin{align}
\sigma^m_{a|x} & = \tr_{A} [( M_{a|x} \otimes \openone) (\tilde{\mu} \tilde{\rho}^{ent}_{pure} + (1- \tilde{\mu}) \tilde{\rho}^{sep}) ] \nonumber \\
& = \tilde{\mu} \tr_{A} [( M_{a|x} \otimes \openone) \tilde{\rho}^{ent}_{pure} ] \nonumber\\
&+ (1 - \tilde{\mu}) \tr_{A} [( M_{a|x} \otimes \openone) \tilde{\rho}^{sep} ] \nonumber \\
& = \tilde{\mu} \tilde{\sigma}^{ent}_{{pure}_{a|x}} + (1 - \tilde{\mu}) \tilde{\sigma}^{sep}_{a|x},
\end{align}
$\tilde{\rho}^{ent}_{pure}$ is a bipartite qubit pure entangled state, $\tilde{\rho}^{sep}$ is a bipartite qubit separable state, $\tilde{\sigma}^{ent}_{{pure}_{a|x}}$ is an element of the assemblage $\{ \tilde{\sigma}^{ent}_{{pure}_{a|x}} \}_{a,x}$ produced from the bipartite qubit pure entangled state $\tilde{\rho}^{ent}_{pure}$ when Alice performs measurements $\{ M_{a|x} \}_{a, x}$ and $\tilde{\sigma}^{sep}_{a|x}$ is an element of the assemblage $\{ \tilde{\sigma}^{sep}_{a|x} \}_{a,x}$ produced from the bipartite qubit separable state $\tilde{\rho}^{sep}$ when Alice performs measurements $\{ M_{a|x} \}_{a, x}$. Since, steerable weight is a convex steering monotone \cite{GA15}
we have
\begin{align}
&SW(\{\sigma^m_{a|x} \}_{a,x}) \nonumber \\
&= SW( \tilde{\mu} \{ \tilde{\sigma}^{ent}_{{pure}_{a|x}} \}_{a,x} + (1 - \tilde{\mu}) \{ \tilde{\sigma}^{sep}_{a|x} \}_{a,x}) \nonumber \\
& \leq \tilde{\mu} SW( \{ \tilde{\sigma}^{ent}_{{pure}_{a|x}} \}_{a,x}) + (1 - \tilde{\mu}) SW(\{ \tilde{\sigma}^{sep}_{a|x} \}_{a,x}), 
\label{eqnew}
\end{align}
where $SW(.)$ denotes the steerable weights of the correspondings assemblages. 
As, $\{ \tilde{\sigma}^{sep}_{a|x} \}_{a,x}$ is the assemblage produced from a separable state, it is an unsteerable assemblage and hence $SW(\{ \tilde{\sigma}^{sep}_{a|x} \}_{a,x}) = 0$ \cite{GA15}. 
On the other hand, $0 \leq SW( \{ \tilde{\sigma}^{ent}_{{pure}_{a|x}} \}_{a,x}) \leq 1$. Hence, from Eq.(\ref{eqnew}) we get
\begin{equation}
SW(\{\sigma^m_{a|x} \}_{a,x}) \leq \tilde{\mu} < 1
\end{equation}
We have, therefore, proved that steerable weight of an arbitrary bipartite qubit mixed entangled state cannot be equal to $1$. On the other hand, it has been shown that the steerable weight of the assemblage produced by performing appropriate measurements on an arbitrary pure entangled state is equal to $1$ \cite{SNC14}.
Since maximum violation of FGI imples that the corresponding assemblage have steerable weight equal to $1$, the maximum violation of FGI is achieved only if the shared bipartite qubit state between Alice and Bob is a pure entangled state.

The general form of any bipartite qubit pure entangled state is given by, $\ket{\psi_{p}}=\ket{\psi(\theta)}=\cos{\theta} \ket{00}+\sin{\theta} \ket{11}$, where $0 < \theta < \frac{\pi}{2}$, up to local unitary transformations. Maximum violation of FGI implies that $P(b_{B_{0}}\mid a_{A_{0}})$ = $1$ and $P(b_{B_{1}}\mid a_{A_{1}})$ = $1$. Let us assume that $a = 0$ and $b = 0$. In this case, it can be easily checked that, for  $B_0$ mentioned above, $P(b_{B_{0}}\mid a_{A_{0}})$ = $1$ for the aforementioned state $\ket{\psi_{p}}$ implies that Alice performs projective measurement of the observable corresponding to the operator
$A_{0}= \ket{0}\bra{0}-\ket{1}\bra{1}$. 
Moreover, it can also be checked that $P(b_{B_{1}}\mid a_{A_{1}})$ = $1$ using aforementioned $B_1$ by the state $\ket{\psi_{p}}$ implies that Alice performs projective measurement of the observable corresponding to the operator $A_{1}=\cos{2\theta}(\ket{0}\bra{0}-\ket{1}\bra{1})+\sin{2\theta}(\ket{+}\bra{+}-\ket{-}\bra{-})$. Hence the claim.

For other possible outcomes ($a$ and $b$) the proof is similar.
\end{proof}

\begin{Lemma}
 The maximal violation of FGI, i.e., $2$ is obtained in our $1$SDI scenario where Bob performs the two mutually unbiased qubit measurements corresponding to the operators given in the previous Lemma $1$.
 Let this maximal violation be achieved from the assemblage arising from the pure state $\ket{\psi}_{AB} \in \mathcal{H}_A \otimes \mathcal{H}_B$
 (where the dimension of $\mathcal{H}_B$ is $2$)
 and measurement operators $\{M_{a|x}\}_{a,x}$ on Alice's side. Then there exists an isometry on Alice's side $\Phi$: $\mathcal{H}_A  \rightarrow \mathcal{H}_A \otimes \mathcal{H}'_A$, where  the dimension of $\mathcal{H}'_A$ is $2$, such that
\begin{align}
 \Phi(\ket{\psi}_{AB})&=\ket{junk}_A \otimes \ket{\psi(\theta)}_{AB}, \nonumber \\
 \Phi(M_{a|x} \otimes \openone \ket{\psi}_{AB})&=\ket{junk}_A \otimes (\tilde{M}_{a|x} \otimes \openone ) \ket{\psi(\theta)}_{AB},
\end{align}
where  $\ket{junk}_A \in \mathcal{H}_A$, $\ket{\psi(\theta)}_{AB}$ is given by Eq. (\ref{state1}) and $\{\tilde{M}_{a|x}\}_{a,x}$ are the measurement 
operators on Alice's side corresponding to the observables given in Lemma $1$.
\end{Lemma}

\begin{proof}
 Let us recall a lemma given in Ref. \cite{Mas06} which states that given two Hermitian operators $A_0$ and $A_1$ with eigenvalues $\pm 1$
 acting on a Hilbert space $\mathcal{H}$, there is a decomposition of $\mathcal{H}$ as a direct sum of subspaces $\mathcal{H}_i$ of dimension $d\le2$ each, such that both $A_0$ and $A_1$ act within
each $\mathcal{H}_i$, that is, they can be written as $A_0= \oplus_i A_0^i$ and
$A_1= \oplus_i A_1^i$, where $A_0^i$ and $A_1^i$ act on  $\mathcal{H}_i$.

 In general, in our steering scenario, any shared bipartite state lies in $\mathcal{B}(\mathcal{H}_A \otimes \mathcal{H}_{B})$ where the dimension of $\mathcal{H}_{A}$ (the untrusted side) is `$d$' and the dimension of $\mathcal{H}_{B}$ (the trusted side) is $2$. 
 From the above lemma it follows that the measurement observables acting on  $\mathcal{H}_{A}$ act within
 each subspace $\mathcal{H}^i_{A}$ with dimension $d \le 2 $ of $\mathcal{H}_{A}$.
 Note that 
\begin{equation}
\mathcal{H}_A \otimes \mathcal{H}_{B} = (\oplus_{i} \mathcal{H}^i_{A}) \otimes \mathcal{H}_{B}
\simeq \oplus_i (\mathcal{H}_{i} \otimes \mathcal{H}_{B}).
\label{HD}
\end{equation}
It follows that the pure state $\ket{\psi}_{AB} \in \mathcal{H}_A \otimes \mathcal{H}_B$
 and the measurement operators $\{M_{a|x}\}_{a,x}$ that give rise to the maximal violation
 of FGI can be decomposed as 
 \be \label{pdec}
 \ket{\psi}_{AB}=\oplus_i \sqrt{q_{i}} \ket{\psi}^i_{AB},
 \ee
 with $\sum_i q_{i}=1$, where $\ket{\psi}^i_{AB}$ is a $2 \times 2$ pure state and
 \be \label{odec}
 M_{a|x}=\oplus_i  M^i_{a|x}, 
 \ee
 where $M^i_{a|x}$ is an operator acting 
 on $\mathcal{H}^i_{A}$ of $d = 2$. 
 
 From our Lemma $1$, it follows that each $\ket{\psi}^i_{AB}$ in Eq. (\ref{pdec}) should be of the following form:
\begin{equation}
\ket{\psi}^i_{AB}=\cos{\theta} \ket{2i,0}+\sin{\theta} \ket{2i+1,1},
\end{equation} 
and for $x=0$, $M^i_{a|x}$  in Eq. (\ref{odec}) are given by $M^i_{0|0}=\ketbra{2i}{2i}$ and $M^i_{1|0}=\ketbra{2i+1}{2i+1}$. 

Alice can append a local ancilla qubit prepared in the state $\ket{0}_A'$ and look for a local isometry $\Phi$ such that 
\be
(\Phi \otimes \openone) \ket{\psi}_{AB}\ket{0}_A'=\ket{junk}_A \otimes \ket{\psi(\theta)}_{A'B},
\ee
where $\ket{junk}_A$ is the junk state and $\ket{\psi(\theta)}_{A'B}$ is the state given by Eq. (\ref{state1}).
This can be achieved for $\Phi$ defined by the map
\ba
\Phi \ket{2k, 0}_{AA'} &\longmapsto&   \ket{2k, 0}_{AA'},\\
\Phi \ket{2k+1, 0}_{AA'} &\longmapsto&   \ket{2k, 1}_{AA'}.
\ea
\end{proof}

Thus, up to local isometry on Alice's side the maximal violation of FGI certifies any pure two-qubit entangled state in our $1$SDI scenario, because any pure two-qubit entangled state can always be written in the form given by Eq. (\ref{state1}) following Schmidt decomposition \cite{Per95, HJW93}. In order to identify which pure entangled two-qubit state has been certified, 
we consider violation of the CFFW inequality (\ref{chshst}).
For the state $\ket{\psi(\theta)}$ given by Eq.(\ref{state1}) with the aforementioned measurement settings on Alice's and Bob's side,  violation of the CFFW inequality is given by,
\begin{equation}\label{CFFWv}
p=\sqrt{4 \sin ^4(\theta )+\sin ^4(2 \theta )}+\sqrt{\sin ^4(2 \theta )+4 \cos ^4(\theta )}.
\end{equation}
Note that concurrence which is a measure of entanglement \cite{CKW00}  of the state $\ket{\psi(\theta)}$ given by Eq. (\ref{state1}) turns out to be $C=\sin{2\theta}$.
It is now readily seen that the violation of CFFW inequality (\ref{chshst}) for the state $\ket{\psi(\theta)}$  given by Eq. (\ref{state1}) and concurrence of this family of states are functions of $\theta$.   Hence, one can easily find out the concurrence of the pure two-qubit state, which is self-tested by the maximum violation of FGI, from the violation ofthe  CFFW inequality. In other words, from the violation of the CFFW inequality one can particularly identify which pure two-qubit entangled state has been self-tested.

The variation of the violation `$p$' with concurrence `$C$' is monotonic and continuous and it is shown in $\figurename\ref{conc_violation_plot}$. Hence, from the violation of the CFFW inequality given by Eq.(\ref{chshst}), one can certify whether the pure two-qubit entangled state is maximally entangled or non-maximally entangled. Moreover, from the violation of CFFW inequality using the plot presented in $\figurename\ref{conc_violation_plot}$, one can find out which particular pure two-qubit entangled state has been self-tested.

\begin{center}
\begin{figure}[!t]
\resizebox{8cm}{4cm}{\includegraphics{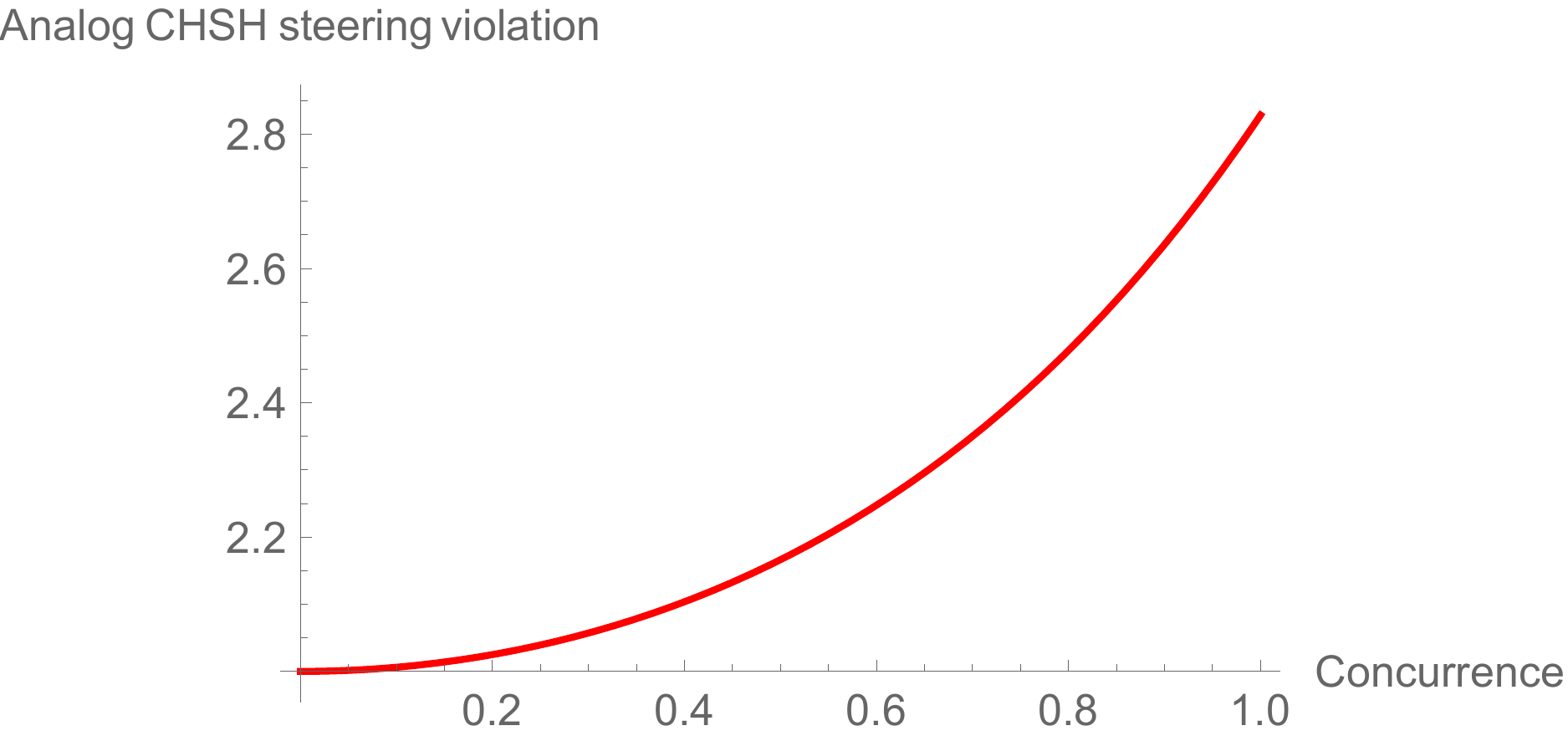}}
\caption{\footnotesize  Along $x$-axis we plot concurrence of the state $\ket{\psi(\theta)}$ and along $y$-axis, we plot the LHS of the analogue CHSH steering inequality.}
\label{conc_violation_plot}
\end{figure}
\end{center}

Therefore, we can state the following self-testing result.

\begin{Result}
The maximal violation of FGI self-tests any pure two-qubit entangled state. On the other hand, magnitude of the violation of CFFW inequality for the measurements that give rise to the maximal violation of FGI certifies the amount of entanglement of the self-tested pure two-qubit entangled state, i.e., magnitude of the violation of CFFW inequality for the measurements that give rise to the maximal violation of FGI identifies which particular pure two-qubit entangled state has been self-tested by FGI.
\end{Result}

Note that for the maximally entangled state, our scheme for $1$SDI self-testing reduces to the $1$SDI self-testing scheme based on the maximal 
violation of the steering inequality given in Eq. (\ref{LSI}).

For the above self-testing scheme, the full knowledge of the measurement correlations $p(ab|xy)$ is needed. We will now propose a scheme which does not
require the full knowledge of the measurement correlations. This scheme is based on the maximal violation of FGI together with non-vanishing value of a 
correlation function called mutual predictability which has been used for constructing entanglement witness \cite{SHB+12} and steering inequality \cite{LLC+18}.  For the dichotomic observables $A_x$ and $B_y$ on Alice's and Bob's side, respectively, the mutual predictability is given by
\be
C_{A_x B_y}=p(a=0,b=0|x,y)+p(a=1,b=1|x,y)
\ee
Let us now consider the following quantity in the context of our $1$SDI scenario.
This quantity is defined in terms of mutual predictability as follows:
\be \label{ent}
E=\min\{C_{A_0B_0},C_{A_1B_1}\}.
\ee
Note that for the pure state given by (\ref{state1}) and the measurements that are specified in Lemma $1$, the above quantity is related to concurrence of the state $\ket{\psi(\theta)}$, given by $C=\sin{2\theta}$, as  $E=\dfrac{1+C^2}{2}$. Hence, the quantity $E$ (or, $2E-1$) is a monotonic function of concurrence of the state $\ket{\psi(\theta)}$. Therefore, we can state another self-testing scheme below..

\begin{Result}
The maximal violation of FGI self-tests any pure two-qubit entangled state. On the other hand, nonvanishing  value of the quantity $2E-1$, where $E$ defined in Eq. (\ref{ent}), for the measurements that give rise to the maximal violation of FGI certifies concurrence of the self-tested pure two-qubit entangled state, i. e., the non-vanishing  value of the quantity $2E-1$ for the measurements that give rise to the maximal violation of FGI identifies which particular pure two-qubit entangled state has been selftested by FGI.
\end{Result}
It may be noted that for the above scheme, 
it is not necessary to assume that the trusted party performs measurements in mutually unbiased bases.  In fact we can assume projective measurements of arbitrary pair of noncommuting qubit observables at trusted party's side. In this case also maximum violation of FGI implies EPR paradox \cite{EPR35}, which is only possible if the shared state is any pure two-qubit entangled state \cite{SharpContra}.
So this scheme can also be used for self-testing in the dimension-bounded steering scenario \cite{MGO+16} where only the Hilbert-space dimension
of measurements of the trusted party is assumed.

\section{Conclusions}
Quantum steering which is a weaker form of quantum inseparabilities compared to Bell-nonlocality, certifies the presence of entanglement 
in a $1$-sided device-independent way. This method for certification of entanglement has  advantages over entanglement 
certification methods based on Bell nonlocality and standard entanglement witnesses. Motivated by this,   
recently,  $1$-sided device-independent self-testing of the maximally entangled two-qubit state 
was proposed by Supic et. al. \cite{SH16} 
and Gheorghiu et. al. \cite{GWK17}   via quantum steering. 
 In this work, we have proposed two schemes to self-test any pure (maximally or non-maximally) bipartite qubit entangled state up to some local unitary, in a $1$SDI way via quantum steering. 
 
 One of our schemes is based on two different steering inequalities, i) Fine-grained steering inequality (FGI) \cite{PKM14} and ii) analogous CHSH inequality for steering \cite{CFF+15}. 
We have shown that the violation of the analogous CHSH inequality for steering together with the maximal violation of FGI self-tests any pure two-qubit
 entangled state. 
 In another scheme, we have demonstrated that the nonvanishing value of a quantity constructed from a correlation function called mutual predictability
 together with the maximal violation of FGI can be used to self-test any pure two-qubit  entangled state in the dimension-bounded
 steering scenario.
 It would be interesting to investigate the robustness of our self-testing schemes in a future study.
 
  \section*{Acknowledgement} 
  CJ thanks Manik Banik and Ashutosh Rai for useful discussions on self-testing and acknowledges S. N. Bose Centre, Kolkata for the
postdoctoral fellowship. DD acknowledges the financial support from University Grants Commission (UGC), Government of India. BB and SS acknowledge the financial support from INSPIRE programme, Department of Science and Technology, Government of India. ASM acknowledges Project No.: DST/ICPS/Qust/2018/98 from Department of Science and Technology, Government of India.

\bibliography{JM}

\end{document}